\documentclass[journal,12pt,draftclsnofoot,onecolumn]{IEEEtran}
\IEEEoverridecommandlockouts
\usepackage{cite}
\usepackage{amsmath,amssymb,amsfonts,amsthm,bm}
\allowdisplaybreaks

\usepackage{algorithm}
 \usepackage{algpseudocode}
\usepackage{balance}
\usepackage{graphicx}
\usepackage{textcomp}
\usepackage{xcolor}
\usepackage[utf8]{inputenc}
\usepackage{flafter}
\usepackage[font=scriptsize,skip=5pt]{caption}
\usepackage{float}
\usepackage{comment}
\usepackage{mathtools}
\usepackage{subcaption}
\usepackage{pgfplots}
\usepackage{ifthen}
\usepackage{tikz}
\usetikzlibrary{calc,math}
\newtheorem{theorem}{Theorem}

\usetikzlibrary{calc,math}

\begin{document}
\author{\IEEEauthorblockN{Rishabh Roy, Rajshekhar V Bhat, Preyas Hathi, Nadeem Akhtar  and Naveen Mysore Balasubramanya}}
\title{Maximization of Timely Throughput with Target Wake Time in IEEE 802.11ax\thanks{The work of Rishabh Roy was supported by Arista Networks. Rishabh Roy, Rajshekhar V Bhat and Naveen Mysore Balasubramanya are with the Indian Institute of Technology Dharwad, Dharwad, India. Preyas Hathi is with Arista Networks and Nadeem Akhtar is with Tarana Wireless Inc.}
}

% \author{\IEEEauthorblockN{Rishabh Roy\IEEEauthorrefmark{1}, Rajshekhar V Bhat\IEEEauthorrefmark{2}, Preyas Hathi\IEEEauthorrefmark{3}, Nadeem Akhtar\IEEEauthorrefmark{4} and Naveen Mysore Balasubramanya\IEEEauthorrefmark{5}}
	
% \IEEEauthorblockA{
% 		\IEEEauthorrefmark{1}\IEEEauthorrefmark{2}\IEEEauthorrefmark{5}Indian Institute of Technology Dharwad, India,    
% 	\IEEEauthorrefmark{3}Arista Networks, \IEEEauthorrefmark{4}Tarana Wireless \\	
% 		Email: \{\IEEEauthorrefmark{1}201082002, 
% 		\IEEEauthorrefmark{2}rajshekhar.bhat, 
% 		\IEEEauthorrefmark{5}naveenmb\}@iitdh.ac.in,  \IEEEauthorrefmark{3}preyas.hathi@arista.com, \IEEEauthorrefmark{4}nakhtar@taranawireless.com} 

\maketitle

\begin{abstract}
 In the IEEE 802.11ax standard, a mode of operation  called target wake time (TWT) is introduced towards enabling deterministic scheduling in  WLAN networks. In the TWT mode, a group of stations (STAs)  can negotiate with the access point (AP) a periodically repeating time window, referred to as TWT Service Period (TWT-SP),  over which they are awake and outside which they \emph{sleep} for saving power. The offset from a common starting time to the first TWT-SP is referred to as the TWT Offset (TWT-O)  and the periodicity of TWT-SP is referred to as the TWT Wake Interval (TWT-WI). %The (TWT-O, TWT-WI) pair completely specifies the time instants over which a group of STAs must wake-up and transmit. 
In this work, we consider communication between multiple STAs with heterogeneous traffic flows and an AP of an IEEE 802.11ax network operating in the TWT mode. 
%, with a given set of possible (TWT-SP, TWT-O, TWT-WI) triplet. 
Our objective is to maximize a long-term weighted average timely throughput across the STAs, where the instantaneous timely throughput is defined  as the number of packets  delivered successfully before their deadlines at a decision instant. To achieve this, we obtain algorithms, composed of (i) an inner resource allocation (RA) routine that allocates resource units (RUs) and transmit powers to STAs,
%resources to STAs in a given group, based on the number packets available in the STAs' buffers and the state of the channel between them and the AP 
and (ii) an outer grouping routine that assigns STAs to (TWT-SP, TWT-O, TWT-WI) triplets. 
% The set of STAs assigned with the same  triplet is considered to belong to a single group and they communicate simultaneously using different RUs.  
%, during their common TWT-SPs, using OFDMA.  
For inner RA, we propose a near-optimal low-complexity algorithm using the drift-plus-penalty (DPP) framework and we adopt a greedy algorithm as outer grouping routine. Via numerical simulations, we observe that the proposed algorithm, composed of a DPP based RA and a greedy grouping routine, performs better than other competitive algorithms.

 \end{abstract}

 \section{Introduction}

 The IEEE Standard IEEE 802.11ax-2021, also known as high-efficiency wireless (HEW) or WiFi-6, is the latest version of the standard for wireless local area networks (WLANs) \cite{Standard}. An important goal of the standard is to enable improved \emph{Quality of Experience} (QoE) of WLAN users %in dense scenarios
 by efficient utilization of spectrum and power resources.  For instance, IEEE 802.11ax expects to improve the average throughput per user station (STA) by a factor of four and reduce latency significantly compared to those in earlier versions, especially  in dense scenarios \cite{Standard,Latency1}.
 
 To enable efficient spectrum utilization, the standard has incorporated  
the use of well-known  orthogonal frequency division multiple access (OFDMA) and multi-user multiple-input and multiple-output (MU-MIMO) for both uplink (UL) and downlink (DL) communications. A technique called spatial reuse, where   new transmissions can be carried out, subject to certain constraints on the transmit power, in presence of ongoing interfering transmissions has also been introduced for better spectrum utilization. Moreover, for efficient use of the  power resource, a mechanism called target wake time (TWT) has been introduced where a group of  STAs can negotiate \emph{wake} (respectively, \emph{sleep}) periods during which the STAs in a group are awake to send or receive data (turn off their transceivers and go to a low-power sleep mode for saving power). 
%\footnote{During this time, the STAs listen only to a specific type of beacon frames called target beacon frame.} 

%In this work, we consider the problem of assigning a given set of TWT parameters, i.e., (TWT-SP, TWT-O, TWT-WI) triplets, to the STAs, along with scheduling and resource allocation to them, in an IEEE 802.11ax network. In the below, we provide a detailed description about TWT, followed by a literature survey on works related to TWT. 
% The goal is to maximize a metric called timely throughput defined as the average number of packets successfully delivered to their destinations within their deadlines. 
  
In this work, we consider optimization of an IEEE 802.11ax network supporting the TWT feature.  
In a TWT session, the wake and sleep cycles for STAs have to be negotiated and scheduled by the participating nodes, namely, APs and STAs. The time duration for which the STA is awake for sending and receiving data is called as 
the \emph{TWT Service Period} (TWT-SP). The time offset after which an STA wakes up following  a target beacon transmission time (TBTT) specified by a special beacon frame called as \emph{Target} beacon 
%https://www.cisco.com/c/en/us/td/docs/wireless/controller/9800/17-2/config-guide/b_wl_17_2_cg/target_wake_time.pdf
 and starts a TWT-SP is called 
as \emph{target wake time}. Both the mechanism and the time duration are referred as Target Wake Time. For avoiding any confusion, we refer to the \emph{target wake time} as the \emph{TWT Offset} (TWT-O).  The time duration between two consecutive TWT-SPs is called as the \emph{TWT wake interval} (TWT-WI). 
With careful selection of TWT parameters, in addition to the power conservation, the TWT mechanism can reduce the contention and the number of collisions significantly as the STAs in the sleep mode do not compete for the common wireless medium \cite{Bellata-TWT-Scheduled-Access,Tutorial}.
\begin{comment}
\footnote{The standard allows other non-WiFi-6 STAs, which are not scheduled under the TWT mechanism, to wake up and compete for channel access during a TWT-SP. We do not consider such scenarios in the current work.}
\end{comment}
Moreover, the parameters of TWT are decided before a TWT session, due to which the TWT provides a deterministic operation and guaranteed service, unlike when the TWT mechanism is absent in which case random access mechanisms  are used.
In this work, we consider the operation of TWT in a mode called the broadcast TWT agreement mode, %where the AP negotiates the TWT parameters with a group of STAs. 
%{\color{red}where an STA (called as the \emph{Scheduled} node) can request to be a member in the existing TWT setup advertised by AP (called as the \emph{Scheduling} node)  or  request creation of a new TWT session with a group of STAs according to some parameters, such as similarity of traffic arrivals or QoE requirements.}
where  the AP assigns an STA a (TWT-O, TWT-WI, TWT-SP) triplet from a set of possible triplets and the STAs with the same (TWT-O, TWT-WI, TWT-SP) triplet form a broadcast TWT group. Within a group, the STAs transmit data using OFDMA, for which we  allocate resource units (RUs) consisting of a specific set of consecutive subcarriers and transmit power. 
In the literature, optimization of various design parameters of IEEE 802.11ax networks supporting the TWT feature have been considered.  For instance, \cite{broadcast-TWT-TBTT-LI-optimization} considers a broadcast TWT setup and proposes a scheme to split STAs into groups by providing each group a different time slot to wake up. 
%by minimizing the number of active STAs in a slot, thus minimizing the contention for random access RUs (RA-RU).
%the authors consider a broadcast TWT setup where STAs are split into groups according to their Listen Interval (LI) and STAs within a group are assigned different NTB such that at a slot, number of active STAs are limited and contention for random access RUs (RA-RU) is minimized.
% They compared the performance of the network in terms of throughput and energy efficiency between the proposed TSS scheme, first come first service (FCFS) scheme, where the NTB of each STA is assigned with the same value of 1 and random (RND) scheme, where the LIs requested by the STAs are all accepted, while the NTB suggested by the AP is randomly selected between 1 and its corresponding LI, and showed that the TSS scheme outperforms FCFS and RND schemes on both throughput and energy efficiency. 
%Authors propose a scheme called TWT based sleep/wake-up scheduling scheme (TSS) which 
%The proposed scheme performs better in terms of total system throughput and energy efficiency compared to other competitive schemes, such as  first-come-first-serve (FCFS) and a  random scheme. The proposed scheme does not consider QoS requirements such as  the delay, TWT-SP resource allocation or optimization of TWT-O and TWT-WI.
%Quality of service (QoS) requirements like delay, different scenarios such as overlapping WLANs, and TWT-SP scheduling and TWT wake interval determination are not considered in this work.
 In \cite{DBLP:Karaca_journals/corr/abs-2006-02235}, the author considers an individual TWT set-up and jointly optimizes TWT wake interval from a predefined set of periodicities and STA scheduling  based on traffic patterns and channel state information, such that the overall energy consumption by the STAs is minimized.
%The proposed algorithm called Joint TWT Interval and Scheduling Algorithm (JTWSA) divides the STAs into groups and assigns shorter  TWT-WI when the buffer occupancy, the number of packets buffered in a STA's buffer, is large  and/or channel condition is good. When the buffer occupancy is low, a longer TWT-WI is chosen so that more sleep time is provided to the STA and hence it can save energy. 
TWT scheduling strategies based on well-known max-rate and proportional fairness schedulers have been considered in \cite{WCNC-9417269}. 
%implicit TWT (TWT-SPs are periodic) for grouped users and established the TWT scheduling strategies by implementing two well-known schedulers, namely,  max-rate (MR) and proportional fairness (PF) schedulers. For the grouping strategy, they have considered STAs with the same MCS value to be in the  same group.
% When the number of STAs is low (two STAs in each group), the TWT does not exhibit any significant improvement compared to the baseline IEEE 802.11 scheme and the request-to-send/clear-to-send (RTS/CTS) scheme in terms of throughput, but
%It is shown that when number of STAs increases (more than 5 STAs per group), TWT with MR scheduler shows higher throughput compared to a case when the network operates in the non-TWT mode. 
%As mentioned, resource allocation within a TWT-SP is an important part of the TWT implementation. 
In \cite{Bellata-Uplink-Resource-Allocation}, the authors consider a utility maximization problem with average rate and power constraints and solve it using the Lyapunov framework.  At each uplink OFDMA transmission, STAs are assigned transmit power and a resource unit where the transmit power is allocated independently in a optimal way over each RU and RU allocation is carried out as a classical assignment problem solved via the Hungarian method. 
%Authors also propose a weighted max-min fairness method for solving the unconstrained utility maximization problem.
% These method can be applied in each TWT-SP for scheduling and resource allocation.

Although considerable work has been carried out for design and optimization of IEEE 802.11ax networks, in the existing literature, optimization of TWT parameters, resource allocation and user scheduling problems have been treated  independently of one another. In a broadcast TWT setup, joint optimization of user grouping, resource allocation at each TWT-SP and TWT parameter selection, especially in a setup where STAs present in the network have delay-bound packets in their buffer has not been explored. Towards exploring this, our  main contributions  are: 
\begin{itemize}
	\item We consider a sum timely throughput maximization problem in  an IEEE 802.11ax network supporting the TWT feature, where fixed but arbitrary number of STAs with delay-bound packets transmit to an AP. 
	%Each station is subjected to an average power constraint and the network supports the target-wake time feature. 
   \item In solving the above problem, we propose an iterative algorithm that consists of inner and outer routines. The inner routine allocates RUs and the transmit powers for a given group of STAs in order to maximize their long-term average timely throughput. The outer routine partitions the set of STAs into pre-specified number of groups, where each group is defined by its associated (TWT-O, TWT-WI, TWT-SP) triplet.   
   %periodicity of transmission, $\rm TWT$-$\rm WI$ {\color{blue} and service period duration, TWT-SP}.
 %  \item For grouping of STAs, we adopt a greedy algorithm and for resource allocation, we adopt {\color{blue} again a greedy algorithm based on Hungarian Algorithm and Drift Plus Penalty based algorithm.} 
   \item Via numerical simulations,  we show that the proposed algorithm, in which the inner RA routine is a drift-plus-penalty based method and the outer grouping routine is a greedy method, performs significantly better  than competitive  algorithms, in which the inner RA routines are either round-robin based or greedy and outer grouping routine  is round-robin based.
\end{itemize}
\section{System Model and Problem Formulation}
\begin{comment}
In this section, we describe the system model considered and formulate an optimization problem for maximizing the long-term average timely throughput. 
\end{comment}
% \begin{figure*}[t]
%     \centering
%     \includegraphics[scale=0.5]{TWT (1).jpg}
%     \caption[An illustration of a broadcast TWT with one AP and one broadcast group.]{An illustration of a broadcast TWT setup with one AP and one broadcast group. The horizontal time axis is divided into epochs which is further divided into slots of equal length. The group shown in the figure has a wake interval of $\tau$ and offset of $o$. }
%     \label{fig:system model}
% \end{figure*}

\begin{figure*}[t]
\centering
\begin{tikzpicture}[scale = 0.55]
\def\xzero{0};
\def\xone{28};
\def\yzero{0};
\def\yone{8};

%% drawing axis
\draw[ultra thick,-{stealth[scale = 4]}] (\xzero,\yzero) -- (\xone,\yzero); 
%\foreach \i in {0,0.5,...,27.5} \draw (\i,-0.25)--(\i,0.25);  % drawing tick marks
\tikzstyle{fontbf} = [scale=0.5,text width=5cm,text centered]
\node[fontbf] (1) at (\xone-0.2,\yzero-0.75) {time};
\draw[ultra thick,-{stealth[scale = 4]}] (\xzero,\yzero) -- (\xzero,\yone-0.5) node[anchor = west]{};
\node[fontbf, text width=2cm] (2) at (\xzero-1.3,\yzero+3.5) {Frequency (as Resource Units)};

% \draw[thick, {stealth[scale = 4]}-{stealth[scale = 4]}] (\xzero+1.5,\yone+1) -- (\xone-1,\yone+1) node[midway, above left = 2 pt]{\textbf{Grouping Interval}};

%% drawing dashed vertical line
\draw[thick,dashed] (\xzero+1.5,\yzero) -- (\xzero+1.5,\yone-1.2);
\draw[thick,dashed] (\xone-2,\yzero) -- (\xone-2,\yone-1.2);
\draw[thick, {stealth[scale = 4]}-{stealth[scale = 4]}] (\xzero+5,\yzero-1) -- (\xzero+7,\yzero-1) node[fontbf, midway,below = 2 pt]{(Block Duration)} node[fontbf,midway,above = 2 pt]{FTT};

\node[fontbf, align=center] at (\xzero+1.5,\yzero+5) {Start of\\ $\bm{i^{th}}$ TWT};
\node[fontbf, align=center] at (\xzero+1.6,\yzero+4.4) {negotiation};
%% TWT negotiation
\node[fontbf, align=center] at (\xone-2,\yzero+5) {Start of\\ $\bm{(i+1)^{th}}$ TWT};
\node[fontbf, align=center] at (\xone-2.1,\yzero+4.4) {negotiation};

%% first frame 
\draw[thick] (\xzero+5,\yzero) rectangle (\xzero+7,\yone-3);  
 \draw (\xzero+5,\yzero+2) -- (\xzero+7,\yzero+2);
 \node[fontbf, align=center] at (\xzero+6,\yzero+1) {STA 3};
 %\draw[dashed] (\xzero+6,\yzero+1) -- (\xzero+6,\yzero+2);
 \draw (\xzero+5,\yzero+3) -- (\xzero+7,\yzero+3);
 \node[fontbf, align=center] at (\xzero+6,\yzero+2.5) {STA 5};
% \draw (\xzero+5,\yzero+4) -- (\xzero+7,\yzero+4);
% \draw[dashed] (\xzero+6,\yzero+4) -- (\xzero+6,\yzero+5);
% \draw (\xzero+5,\yzero+5) -- (\xzero+7,\yzero+5);
% \node[fontbf, align=center] at (\xzero+6,\yzero+5.5) {STA 7};
% \draw (\xzero+5,\yzero+6) -- (\xzero+7,\yzero+6);
\node[fontbf, align=center] at (\xzero+6,\yzero+4) {STA 6};

%% naming the frame
\node[fontbf,align=center] at (\xzero+6,\yzero+5.2) {Frame \#1};

%% Between two frames
\draw[thick, dashed] (\xzero+7,\yzero+2.5) -- (\xzero+9,\yzero+2.5);
\node[fontbf, align=center] at (\xzero+8,\yzero+3.5){multiple\\Frame\\Transmission};

%% Next frame
\draw[thick] (\xzero+9,\yzero) rectangle (\xzero+11,\yone-3);
\draw (\xzero+9,\yzero+2) -- (\xzero+11,\yzero+2);
\node[fontbf, align=center] at (\xzero+10,\yzero+1) {STA 5};
 \draw (\xzero+9,\yzero+3) -- (\xzero+11,\yzero+3);
 \node[fontbf, align=center] at (\xzero+10,\yzero+2.5) {STA 6};
% \draw[dashed] (\xzero+10,\yzero+2) -- (\xzero+10,\yzero+4);
% \draw (\xzero+9,\yzero+4) -- (\xzero+11,\yzero+4);
% \node[fontbf, align=center] at (\xzero+10,\yzero+4.5) {STA 1};
% \draw (\xzero+9,\yzero+5) -- (\xzero+11,\yzero+5);
% \node[fontbf, align=center] at (\xzero+10,\yzero+5.5) {STA j};
% \draw (\xzero+9,\yzero+6) -- (\xzero+11,\yzero+6);
 \node[fontbf, align=center] at (\xzero+10,\yzero+4) {STA 3};

%% naming the frame
\node[fontbf,align=center] at (\xzero+10,\yzero+5.2) {Frame \#n};

%% TWT SP 
\draw[thick, {stealth[scale = 3]}-{stealth[scale =3]}] (\xzero+5,\yzero+5.8) -- (\xzero+11,\yzero+5.8) node[fontbf, midway,below = 2 pt]{(TWT-SP)} node[fontbf,midway,above = 2 pt]{${\zeta}$ blocks};

%% TWT O
\draw[thick, {stealth[scale = 3]}-{stealth[scale =3]}] (\xzero+1.5,\yone-1.5) -- (\xzero+5,\yone-1.5) node[fontbf, midway,below = 2 pt]{(TWT-O)} node[fontbf,midway,above = 2 pt]{{o} blocks};

%% TWT WI

\draw[thick, {stealth[scale = 3]}-{stealth[scale =3]}] (\xzero+5,\yone-1.5) -- (\xzero+16,\yone-1.5) node[fontbf, midway,below = 2 pt]{(TWT-WI)} node[fontbf,midway,above = 2 pt]{${\tau}$ blocks};
%\draw (\xzero+9,\yzero+7) -- (\xzero+11,\yzero+7);
%\node[align=center] at (\xzero+10,\yzero+7.5) {STA $$};
% \draw (\xzero+5,\yzero+4) -- (\xzero+7,\yzero+4);
% \draw[dashed] (\xzero+6,\yzero+4) -- (\xzero+6,\yzero+5);
% \draw (\xzero+5,\yzero+5) -- (\xzero+7,\yzero+5);
% \node[align=center] at (\xzero+6,\yzero+5.5) {STA $7$};
% \draw (\xzero+5,\yzero+6) -- (\xzero+7,\yzero+6);
% \node[align=center] at (\xzero+6,\yzero+7) {STA $5$};

%% Between Two SPs
\draw[thick, dashed] (\xzero+11,\yzero+2.5) -- (\xzero+16,\yzero+2.5);
\node[fontbf, align=center] at (\xzero+13.5,\yzero+3.5){TWT-SP\\Of Other\\Broadcast Groups};

%% Next SP 1st frame
\draw[thick] (\xzero+16,\yzero) rectangle (\xzero+18,\yone-3);
\draw (\xzero+16,\yzero+2) -- (\xzero+18,\yzero+2);
\node[fontbf, align=center] at (\xzero+17,\yzero+1) {STA 6};
% \draw[dashed] (\xzero+17,\yzero+1) -- (\xzero+17,\yzero+3);
\draw (\xzero+16,\yzero+3) -- (\xzero+18,\yzero+3);
 \node[fontbf, align=center] at (\xzero+17,\yzero+2.5) {STA 3};
% \draw (\xzero+16,\yzero+4) -- (\xzero+18,\yzero+4);
% \node[fontbf, align=center] at (\xzero+17,\yzero+4.5) {STA 7};
% %\draw[dashed] (\xzero+17,\yzero+4) -- (\xzero+17,\yzero+5);
% \draw (\xzero+16,\yzero+5) -- (\xzero+18,\yzero+5);
% \node[fontbf, align=center] at (\xzero+17,\yzero+5.5) {STA 3};
% \draw (\xzero+16,\yzero+6) -- (\xzero+18,\yzero+6);
 \node[fontbf, align=center] at (\xzero+17,\yzero+4) {STA 5};

%% naming the frame
\node[fontbf,align=center] at (\xzero+17,\yzero+5.2) {Frame \#1};

%% Between frames
\draw[thick, dashed] (\xzero+18,\yzero+2.5) -- (\xzero+20,\yzero+2.5);
\node[fontbf, align=center] at (\xzero+19,\yzero+3.5){multiple\\Frame\\Transmission};

%% Next Frame
\draw[thick] (\xzero+20,\yzero) rectangle (\xzero+22,\yone-3);
\draw (\xzero+20,\yzero+2) -- (\xzero+22,\yzero+2);
\node[fontbf, align = center] at (\xzero+21,\yzero+1) {STA 3}; 
\draw (\xzero+20,\yzero+3) -- (\xzero+22,\yzero+3);
\node[fontbf, align = center] at (\xzero+21,\yzero+2.5) {STA 6}; 
% \draw[dashed] (\xzero+21,\yzero+3) -- (\xzero+21,\yzero+5);
% \draw (\xzero+20,\yzero+5) -- (\xzero+22,\yzero+5);
% \node[fontbf, align = center] at (\xzero+21,\yzero+6) {STA j}; 
% \draw (\xzero+20,\yzero+7) -- (\xzero+22,\yzero+7);
\node[fontbf, align = center] at (\xzero+21,\yzero+4) {STA 5}; 

%% naming the frame
\node[fontbf,align=center] at (\xzero+21,\yzero+5.2) {Frame \#n};

%% After Next SP
\draw[thick, dashed] (\xzero+22,\yzero+2.5) -- (\xzero+26,\yzero+2.5);
\node[fontbf, align=center] at (\xzero+24,\yzero+3.2){multiple\\TWT-SPs};

\end{tikzpicture}
\caption{An illustration of a broadcast TWT setup where one AP and multiple STAs belonging to several broadcast groups are present in the BSS. The horizontal time axis is divided into slots/blocks of length equal to the frame transmission time (FTT) which is assumed to be same across all the STAs. The figure shows a broadcast TWT group consisting of STA $3$, STA $5$ and STA $6$ with the TWT-O ($o$ blocks) TWT-WI ($\tau$ blocks), and TWT-SP ($\zeta$ blocks). A TWT-SP can consist of multiple IEEE 802.11ax frame transmissions of lengths equal to FTT.}
\label{fig:broadcastTWTFigure}
\end{figure*}
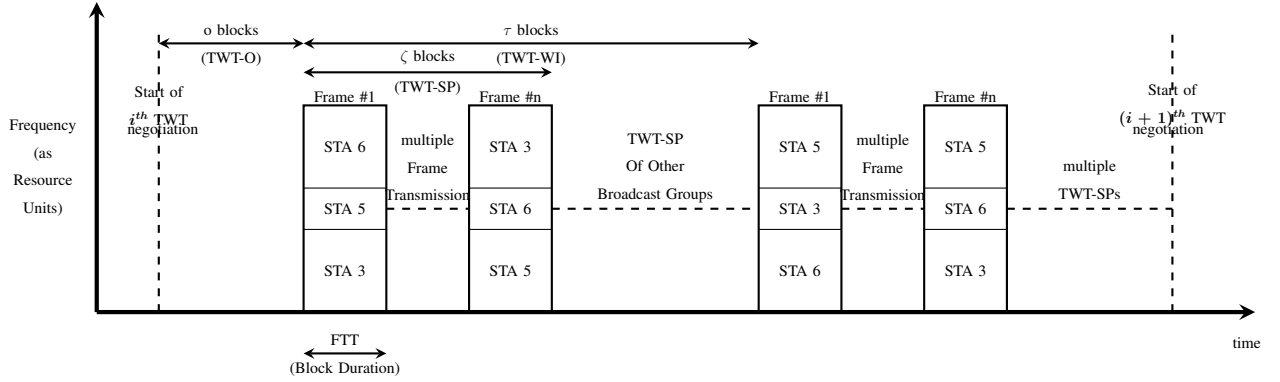

\subsection{System Model}
We consider an IEEE 802.11ax network's infrastructure basic service set (BSS),
% \footnote{A BSS can be divided into two categories: (i) Independent BSS and (ii) Infrastructure BSS. In the independent BSS, STAs directly communicate with each other without the presence of a central entity whereas in the latter, communication is maintained through a centralized entity called access point (AP) and STAs which are placed within the coverage area of the AP belong to the BSS \cite{matthew-gast}. Due to relevance to  the current paper, we only discuss about the  infrastructure BSS and call it as BSS.},
where an AP is serving $M$ STAs. We consider that the STAs are indexed by the entries in $\mathcal{M}\triangleq\{1, 2, \ldots, M\}$. Similarly, the RUs available are indexed by the entries in $\mathcal{U}\triangleq\{1, 2, \ldots, U\}$. We consider that time is slotted and the length of each slot/block is equal to the  frame transmission time  (FTT), where a frame consists of preambles followed by the payload data. We assume that the length of the preamble part of a frame is negligible. The blocks (slots)  are indexed by $t\in \{1,2,\ldots,\}$. 
\begin{comment}
We assume that the FTT is same across all the STAs and AP in the network\footnote{In practice, the FTT can vary across STAs. However, for simplicity of presentation, we assume them to be same for all STAs.}. 

All the duration of time considered in the paper are expressed as integer multiples of the FTT.
\end{comment}
We consider a broadcast TWT set up where an STA can be assigned to one of the $L$ TWT groups, indexed by $l\in \mathcal{L}\triangleq\{1,2,\ldots,L\}$.  
%, where each group is characterized by the periodicity with which the STAs assigned to the group must transmit. The maximum number of STAs in a TWT group is considered to be $F$. 
%Slots are not necessarily of equal lengths and duration of slots may vary from one slot to another. 
%The duration of each slot is $\theta$ milliseconds. 
The following parameters related to broadcast TWT are important for the problem formulation (See Fig.~\ref{fig:broadcastTWTFigure}).

\begin{itemize}
    \item {TWT Session or Service Period (SP):} It is the time during which STAs from a TWT group wake up and transmit/receive data. A TWT-SP can consist of multiple IEEE 802.11ax frame transmissions. We denote the TWT SP duration associated with group $l \in \mathcal{L}$ by $\zeta_{l}$ blocks. 
    %A TWT-SP can consist of multiple IEEE 802.11ax frame transmissions, where a frame consists of preambles followed by the payload data.
    %TWT SP of different broadcast TWT groups may or may not have equal duration and denoted by $\zeta_{l}$, where  $l \in \mathcal{L}\triangleq \{1,2,\dots,L\}$ denotes the TWT group index. There can be multiple frame transmissions/receptions inside a TWT SP depending on the allocated OFDMA frame transmission time for that broadcast TWT group. 
    %Both TWT SP duration and OFDMA frame transmission time is expressed as an integer multiple of slots.
    
    %where each slot consists of one OFDMA frame transmission duration and each TWT SP consists of multiple such slots.}
    % It may or may not be same for all groups and is a multiple of slot length $\tau$ defined later. 
    \item {TWT wake interval (TWT-WI) or group periodicity:} It is the time difference between two consecutive TWT-SPs of a broadcast TWT group. Broadcast TWT-SPs are periodic and all broadcast TWT groups have a non-zero period, which we denote as $\tau_l$ blocks,  for all $l \in \mathcal{L}$. 
%    For the current work, the values of $\tau_1,\tau_2,\ldots, \tau_L$ are considered to be given. 
    \item {Target Wake Time/Offset (TWT-O):} It is the offset from the TWT negotiation instant after which group $l$ has its first TWT-SP. We denote it by $o_l$ blocks. 
\end{itemize}

 %{\color{red}In the current work, we assume that a frame length is equal to the TWT-SP length.}

%  This implies that the STAs in group $l$ transmit during blocks indexed by $o_l+n\tau_l$ for $\zeta_{l}$ number of blocks, for $n\in \{0,1,\ldots\}$.  In this work, we only consider a class of non-overlapping groups, where the TWT parameters $o_l$, $\tau_l$ and $\zeta_l$ are such that $\{(o_j+n\tau_j, o_j+n\tau_j+\zeta_j)\}_{n=0}^{\infty}\cap\{(o_k+n\tau_k, o_k+n\tau_k+\zeta_k)\}_{n=0}^{\infty} = \phi$ for  any $j\neq k$ and $j,k\in \{1,2,\ldots,L\}$. 
%   At each block, we assume that delay-bound packets arrive at STAs, which are required to be delivered to the AP over a wireless channel.
  
  In the below, we describe  packet arrival and channel models. 

\subsubsection{Data Arrival and Queue Evolution}
Let $a_m(t)$ be the number of packets arriving at the start of block $t$ in the buffer of STA $m\in \mathcal{M}$. We consider a bounded $a_m(t)$, where $0\leq a_m(t)\leq A_{\rm max}$ for some positive integer $A_{\rm max}>0$ for all $m\in \mathcal{M}$ and $t\in \{1,2,\ldots\}$. 
% We consider that data arrives at an STA's buffer as packets and the packets of STA $m$ have the payload of  length $\beta_m$ bits and that a negligible number of bits are dedicated for header and other control information. 
The packets of STA $m$ have the payload of length $\beta_m$ bits and associated deadline of $\delta_m$ blocks, after which they are discarded. Moreover, the  buffer of STA $m$ can store $B_m^{\rm max}$ number of packets.
%(of length $b_m$ bits).
Let $B_m(t)$ be the number of packets at the start of block $t$ in the buffer of STA $m$. 
The buffer state, $B_m(t)$ evolves as follows: 
\begin{align} 
&B_m(t+1)= \max\{B_m(t) +a_m(t)- R_m(t)-d_m(t), 0\}, \label{eq:buffer evolution}
\end{align}
where $R_m(t)$ and $d_m(t)$ are the number of unexpired packets transmitted to the AP and the number of packets discarded at STA $m$, respectively, in block $t$. 
Physically, the buffer operates as follows: At the start of block $t$, when $a_m(t)$ number of packets arrive, if $B_m(t) +a_m(t)>B_m^{\rm max}$, then the oldest packets in the buffer are discarded (ties are broken randomly). After  transmitting $R_m(t)$ number of packets by the end of block $t$, if some packets expire, all the bits belonging to them are discarded. 
The quantity $d_m(t)$ accounts for the number of packets discarded both due to unavailability of buffer space at the start of $t^{\rm th }$ block when new packets arrive and  time expiry at the end of the block when old packets expire. 
%Let the number of slots left before the expiry of the  packet at the  $i^{\rm th}$ location of  the buffer of STA $m$ be $\delta_{m,i}(t)$.

%Note that $R_m(t)$ will be the timely throughput in slot $t$ as all the bits that are transmitted are within their expiry time. This is because, when expired, the packet and all its corresponding bits are discarded, before making a decision in the current slot. 

\subsubsection{Channel Model}
Let $h_{m,k}(t)\in \mathcal{H}$ be the channel power gain of STA $m$ on RU $k$ in block $t$ for $m\in\mathcal{M}$ and $k\in\mathcal{U}$, where $\mathcal{H}$ is a finite subset of  the set of positive real numbers, $\mathbb{R}^+$.  The bandwidth of RU $k$ is $b_{k}$ units. 
We assume that the channel power gains on every RU of a STA remains constant over a block, after which it varies randomly in an independent and identical distributed manner, as in \cite{Bellata-Uplink-Resource-Allocation}.

%Finally, let STA $m$ transmit on RU $k$ with power $P_{m,k}(t)\in \mathcal{P}$, where $\mathcal{P}$ is a finite subset of $\mathbb{R}^+$. 

\subsection{Problem Formulation}
We consider that the maximum  rate of reliable communication
achieved by STA $m$ when RU $k$ is allocated to it and when  transmission happens with power $P_{m,k}(t)$ is: $r_{m,k}(t) = \rho(P_{m,k}(t),h_{m,k}(t), b_k)$ bits per second. Let $u_{m,k}(t)$ be the indicator variable that takes value $1$ if the RU $k\in \mathcal{U}$ is allocated to STA $m\in \mathcal{M}$ and zero otherwise, in block $t\in \{1,2,\ldots,\}$.
% An RU can be allocated to at most one STA and vice versa, i.e., we must satisfy  $\sum_{k=1}^{U}u_{m,k}(t) \leq 1$  for all  $m\in \mathcal{M}$ and $\sum_{m=1}^{M}u_{m,k}(t) \leq 1$ for all $k\in \mathcal{U}$.
In this case, the maximum achievable timely throughput in block $t$ is given by 
\begin{align}\label{eq:ImmediateTimelyThroughput}
R_m(t) = &\min\Bigg\{\left \lfloor{\frac{\mathrm{FTT}}{\beta_m}\times\sum_{k=1}^{U}u_{m,k}(t)r_{m,k}(t)}\right\rfloor, B_m(t)+a_m(t),B_m^{\rm max}\Bigg\},
\end{align} 
where $\lfloor x \rfloor$ is the greatest integer less than or equal to $x$. The long-term average timely throughput for STA $m$ is defined as:  
\begin{align}\label{eq:TimelyThroughput}
\bar{R}_m = \lim_{T\rightarrow \infty}\frac{1}{T}\sum_{t=1}^{T}\mathbb{E}[R_m(t)], 
\end{align}
and the long-term average power is defined as the 
\begin{align}\label{eq:TransmitPower}
\bar{P}_m = \lim_{T\rightarrow \infty}\frac{1}{T}\sum_{t=1}^{T}\mathbb{E}\left[P_m(t)\right], 
\end{align}
where $P_m(t) = \sum_{k=1}^{U}u_{m,k}(t)P_{m,k}(t)$. The expectation in \eqref{eq:TimelyThroughput} and \eqref{eq:TransmitPower} are with respect to the randomness in the  timely throughput and transmit power induced by the  randomness in the channel power gain, and the transmission policy which decides the grouping of the STAs and resource allocation to them.  We note that both the long-term average timely throughput and the long-term average power depend on the broadcast TWT parameters, $o_l$, $\tau_l$ and $\zeta_l$. This is because the STA $m$ obtains the transmission opportunity for its uplink transmission only at the blocks dictated by broadcast TWT parameters as mentioned earlier. 
% However, for brevity, we do not explicitly parameterize the expressions for $\bar{R}_m$ and $\bar{P}_m$ with the TWT parameters. 

%(to be precisely defined at the end of the section). 
%\emph{It is important to note that the timely throughput expression in \eqref{eq:TimelyThroughput} appears similar to the usual long-term average throughput expression (for instance, see the throughput expression in \cite{Bellata-Uplink-Resource-Allocation}). However, the expression in \eqref{eq:TimelyThroughput} is indeed the timely throughput as no expired packet is kept in the buffer.}

As mentioned, our goal is to maximize the long-term weighted average timely throughput subject to average power constraints at STAs. This can be achieved by solving following optimization problem: 
\begin{subequations}\label{eq:Main Problem}
	\begin{flalign}
	\underset{\substack{u_{m,k}(t)\in \{0,1\},\\ P_{m,k}(t)\in \mathcal{P},\mathcal{G}_l}}{\text{maximize}} &\;\;\sum_{m=1}^{M}\alpha_m\bar{R}_m, &&\\
	\text{subject to}\; &\sum_{k=1}^{U}u_{m,k}(t) \leq 1,  \forall   m\in \mathcal{M},&&\label{eq:OneRUtoAtMostOneUser}\\
&\sum_{m=1}^{M}u_{m,k}(t) \leq 1,  \forall  k\in \mathcal{U},\label{eq:OneUsertoAtMostOneRU}&&\\	
& \sum_{k=1}^{U}u_{m,k}(t)=0,\forall t\notin o_l+n\tau_l,o_l+n\tau_l+1,\ldots,o_l+n\tau_l+\zeta_l,\;\forall m\in \mathcal{G}_l, n=0,1,\ldots,\label{eq:PeriodicityConstraint}&&\\
& \bar{P}_m \leq P_{m,\text{avg}}, \forall m\in \mathcal{M},\label{eq:averagePowerConstraint}&&\\
& P_{m,k}(t) \leq P_{\text{max}},  \forall m\in \mathcal{M},  k\in \mathcal{U}, \label{eq:maximumPowerConstraint}&&\\
&\mathcal{G}_v\cap \mathcal{G}_u =\phi,\; \cup_{l=1}^L\mathcal{G}_l=\mathcal{M}, \forall u, v\in \mathcal{L}, &&\label{eq:GroupSizeConstraint}
	\end{flalign}
\end{subequations} 
for all $t\in \{1,2,\ldots\}$ and $l\in \mathcal{L}$, where $\alpha_m\geq 0$ is a given weight parameter for scaling the timely throughput of STA $m$, $\mathcal{G}_l$ is the set of STAs in group $l\in \mathcal{L}$ and $\mathcal{P}$ is a finite subset of $\mathbb{R}^+$. 
The constraint \eqref{eq:OneRUtoAtMostOneUser} says that an RU can be assigned at most to one user in a slot and \eqref{eq:OneUsertoAtMostOneRU} says that an STA can be assigned at most one RU in a slot. The constraint \eqref{eq:PeriodicityConstraint} accounts for the fact that if  STA $m$ belongs to the group $\mathcal{G}_l$, then it must transmit only in the $\eta_l$ consecutive blocks with indices having periodicity $\tau_l$, with $o_l$ being the first slot.  The constraint \eqref{eq:averagePowerConstraint} says that the  long term average transmit power of STA $m$ must be less than or equal to a given value ${P}_{m,\text{avg}}$ and \eqref{eq:maximumPowerConstraint} says that instantaneous transmit power of STA $m$ in RU $k$ must be less than or equal to $P_{\rm max}$, as per the WLAN 802.11ax standard. Finally, \eqref{eq:GroupSizeConstraint} says that 
%the maximum number of STAs in a group must be less than or equal to $F$ and that 
the groups must be mutually exclusive and collectively exhaustive.

\section{Solution}\label{sec:Solution}
The optimization problem in \eqref{eq:Main Problem} is an integer program with decision variables, $P_{m,k}(t)\in \mathcal{P}$ and $u_{m,k}(t)\in \{0,1\}^{M\times U}$ for $m\in \mathcal{M}$ and $k\in\mathcal{U}$ which depend on the time index, $t\in \{1,2,\ldots,\}$ and $G_1,\ldots, G_L$ which do not depend on the time index. The number of different ways in which RUs and  transmit powers can be allocated and groups can be formed is exponential. Hence, solving \eqref{eq:Main Problem} directly is prohibitive. 
Hence, in this section, we propose several low complexity algorithms for solving \eqref{eq:Main Problem}, where an algorithm consists of an inner resource allocation (RA) routine, which optimizes the decision variables that depend on time, and an outer grouping routine which optimizes decision variables which do not depend on time. That is, the RA routine  assigns RUs and transmit powers to all the STAs in a group in order to maximize the {long-term average expected sum timely throughput} of STAs in the group. The outer grouping routine  groups the STAs such that the total long-term average timely throughput across all the STAs are maximized.  
Specifically, let $f_{l}^{*}(\mathcal{G})$ be the maximum long-term average expected timely throughput for a group of STAs, $\mathcal{G}\subseteq \mathcal{M}$, assigned with $o_l$, $\tau_l$, and $\zeta_l$ as  TWT parameters, i.e.,  
%\begin{subequations}
	\begin{align}\label{eq:RAproblem}
	f_{l}^*(\mathcal{G}) = \underset{\pi\in \Pi}{\text{max}} \;\;\sum_{m\in \mathcal{G}}\alpha_m\bar{R}_m, 
	\;\text{subject to}\;\; \eqref{eq:OneRUtoAtMostOneUser}-\eqref{eq:maximumPowerConstraint}.
	\end{align}
%\end{subequations}	
% 	&\sum_{k=1}^{U}u_{m,k}(t) \leq 1,  \forall   m\in \mathcal{G}, t\in \{1,2,\ldots\},\label{eq:OneRUtoAtMostOneUser1}\\
% 	&\sum_{m\in \mathcal{G}}u_{m,k}(t) \leq 1,  \forall  k\in \mathcal{U}, t\in \{1,2,\ldots\},\label{eq:OneUsertoAtMostOneRU1}\\	
% 	& \sum_{k=1}^{U}u_{m,k}(t)=0,\forall t\notin \{o+n\tau\}_{n=0}^{\infty}, m\in \mathcal{G},\label{eq:PeriodicityConstraint1}
% 	\end{align}
% \end{subequations} 
Let $f_{l}^{\pi}(\mathcal{G})$ denote the objective value of \eqref{eq:RAproblem} under a feasible policy $\pi$, which gives a rule to  assign RUs and transmit powers to STAs. 
Given \(f_{l}^{\pi}(\mathcal{G})\), the grouping problem is cast as the following optimization problem:
\begin{subequations}\label{eq:Groupingproblem}
	\begin{align}
	R^{\Sigma} =  \underset{\mathcal{G}_1,\ldots,\mathcal{G}_L}{\text{max}} &\;\;\sum_{l=1}^{L}f_{l}^{\pi}(\mathcal{G}_l) , &&\\
	\text{subject to}&\;\; \mathcal{G}_v\cap \mathcal{G}_u =\phi,\; \forall u, v\in \mathcal{L}, \cup_{l=1}^L\mathcal{G}_l=\mathcal{M}. 
	\end{align}
\end{subequations} 
Let $R^{\Sigma}_{\psi,\pi}$ be the objective value of \eqref{eq:Groupingproblem} when the policy $\pi$ is used to obtain $f_{l}^{\pi}(\mathcal{G})$ and  routine $\psi$ is used for solving \eqref{eq:Groupingproblem}. 

Below, we propose several routines to solve \eqref{eq:RAproblem} and  \eqref{eq:Groupingproblem}.

%In this work, we propose $\Pi$ to be the class of stationary deterministic policies and $\Lambda$ to be the class of greedy policies. 

%\section{Solution to the Resource Allocation Problem in \eqref{eq:RAproblem}}

\subsection{Solution to the Resource Allocation (RA) Problem
in \eqref{eq:RAproblem}}\label{eq:RAsolution}
\begin{comment}
In this subsection, we propose several routines that solve the RA problem. We first present a drift-plus-penalty based routine which provides a near-optimal solution to  \eqref{eq:RAproblem}. \end{comment}

% \subsubsection{Drift-Plus-Penalty (DPP) RA Routine} \label{sec:DPP RA Algorithm}
\begin{comment}
Solving  \eqref{eq:RAproblem} is challenging as the optimization needs to be carried out across multiple decision instants dynamically based on the buffer and channel states while satisfying average power constraints at the STAs. 
\end{comment}
\subsubsection{Drift-Plus-Penalty (DPP) RA Routine} \label{sec:DPP RA Algorithm}
In the DPP method, we derive an optimization problem from  \eqref{eq:RAproblem} that can be greedily solved at every block, leading to a near-optimal performance. 
In the DPP method, 
% the constraints are accounted for by treating them as virtual queues and ensuring their stability.  Concretely,
we consider that the virtual queue corresponding to the long-term average power constraint in STA $m\in \{1,\ldots,M\}$, $\bar{P}_{m} \leq P_{m,\text{avg}}$  evolves as follows:   
\begin{align}\label{eq:virtualqueue}
	G_{m}(t+1) = \max\{G_{m}(t)-P_{m,\text{avg}}+P_m(t),0\},
\end{align} 
for all $t\in \{1,2,\ldots\}$, where $G_{m}(t)$ is the virtual queue backlog at block $t$, $P_{m,\text{avg}}$ corresponds to the virtual queue service rate, and $P_m(t)$ is the virtual queue arrival process.
Let $\Theta(t)= \{B_m(t),G_m(t), \mathbf{h}_{m}(t)\}_{m=1}^M$  be the state of the system at block $t$, where $\mathbf{h}_{m}(t)=(h_{m,k}(t))_{k=1}^{U}$. We now consider a quantity called as the conditional Lyapunov drift defined as:
\begin{align}\label{eq:drift}
	D(\Theta(t))   = \mathbb{E}[L(\Theta(t+1))  - L(\Theta(t))| \Theta(t)].  
\end{align}
where $L(\Theta(t)) = \frac{1}{2}\sum_{m=1}^{M}B_m^2(t) + \frac{1}{2}\sum_{m=1}^{M}G_m^2(t)$ is called as a quadratic Lyapunov function. 
%$	L(\Theta(t)) = (1/2)\sum_{m=1}^{M}B_m^2(t)$ and the following drift-plus-penalty optimization problem is formulated: 
Consider the following: 
\begin{equation} \label{DPP}
	\begin{aligned}
		\text{minimize} \;\;&  D(\Theta(t)) - V\sum_{m=1}^{M}\mathbb{E}[R_m(t)|\Theta(t)], \\
 		\text{subject to} \;\;&\eqref{eq:OneRUtoAtMostOneUser}-\eqref{eq:maximumPowerConstraint},
    \end{aligned}
\end{equation}
%\sum_{k=1}^{U}u_{m,k}(t) \leq 1,  \forall   m\in \mathcal{G}, t\in \{1,2,\ldots\}, \\
% 		&\sum_{m\in \mathcal{G}}u_{m,k}(t) \leq 1,  \forall  k\in \mathcal{U}, t\in \{1,2,\ldots\}, \\
% 		& \sum_{k=1}^{U}u_{m,k}(t)=0,\forall t\notin \{o+n\tau\}_{n=0}^{\infty}, m\in \mathcal{G},
% 	\end{align}
% \end{subequations}
for all $m \in \{1,2,\dots, M\}$, $t \in {\{1,2,\dots\}}$. Here, $V>0$ is the parameter that trades-off the importance of the total conditional expected timely throughput (negative of which is considered as the penalty) and the drift term in \eqref{eq:drift}.
It is challenging to compute the objective function of \eqref{DPP} and solve it.
% Hence, we replace the objective function in \eqref{DPP} by its upper bound, solve the problem and obtain a near-optimal performance. We provide a detailed discussion of DPP method and derivation of the  upper bound in the Appendix in \cite{longer-version}.
Hence, we consider the following upper bound to the objective function in \eqref{DPP}: 
\begin{align}\label{eq:DriftPlusPenalty}
	&D(\Theta(t)) - V\sum_{m=1}^{M}\mathbb{E}[R_m(t)|\Theta(t)]\leq   C + \mathbb{E}[U_B|\Theta(t)],
\end{align}
where  \(C = \frac{1}{2}\sum_{m=1}^{M} (P_{m,\text{avg}}^2 + P_{\rm max}^2+  B_{\rm max}^2 + A_{\rm max}^2)\) and $U_B = \sum_{m=1}^{M} G_m(t)(P_m(t)-P_{m,\text{avg}})+\sum_{m=1}^{M}B_m(t)(a_m(t)-R_m(t)-d_m(t)-V\sum_{m=1}^{M}R_m(t)$. We provide a detailed derivation of the  upper bound in \eqref{eq:DriftPlusPenalty} in Appendix A.

	\begin{algorithm}[b]
		\caption{The DPP-RA Routine.}
		\label{algo:DPP-RA}
		{\small
			\begin{algorithmic}[1]
				\Procedure{DPP-RA}{}
                 \For{$t\in \{1,2,\ldots,T\}$}
				\State Obtain $h_{m,k}(t)$, and update $B_m(t)$ and $G_m(t)$ according to  \eqref{eq:buffer evolution} and \eqref{eq:virtualqueue}, respectively. 
    \State For $k\in \mathcal{U}$ and STA $m\in \mathcal{M}$, obtain $P_{m,k}^*(t)$ that maximizes $U_B$. Denote the maximum value by $U_B^*$.  
    \State Construct  $A\in \mathbb{R}^{M \times U}$, whose $(i,j)^{\rm th}$ entry is $U_B^*$ when $j^{\rm th}$ RU is allocated to the $i^{\rm th}$ STA.
   \State Obtain the optimal assignment, $u_{m,k}(t)^*$, by applying the Hungarian (Munkres) algorithm on $A$.
%    \State Increment $t$. Update the state of the buffer and virtual queues according to , respectively. Go to $(2)$ until the maximum number of decision blocks is reached. 
          \EndFor
				\EndProcedure
		\end{algorithmic}}
	\end{algorithm}
The  DPP method observes $\Theta(t)$ at each decision slot $t$ and chooses actions $u_{m,k}(t)$ and $P_{m,k}(t)$, $\forall m, \forall k$, to minimize the right hand side of \eqref{eq:DriftPlusPenalty}. Noting that the conditional expectation in \eqref{eq:DriftPlusPenalty} is minimized when each realization of $U_B$ minimized, we omit the expectation and solve the following:
\begin{equation} \label{DPPGreedyOptimization}
    \begin{aligned}
 		\underset{\substack{u_{m,k}(t)\in \{0,1\},\\ P_{m,k}(t)\in \mathcal{P},}}
   {\text{minimize}} \;\;  U_B,\;\; 
		\text{subject to} \;\;\eqref{eq:OneRUtoAtMostOneUser}, \eqref{eq:OneUsertoAtMostOneRU},\eqref{eq:maximumPowerConstraint},
	\end{aligned}
\end{equation}
for each $t \in {\{1,2,\dots\}}$.  
 Note that the transmit powers are allocated from a finite set of available power levels, $\mathcal{P}$, and the number of available RUs is finite, which is given by the set, $\mathcal{U}$. Hence, \eqref{DPPGreedyOptimization} is an integer program. We solve it optimally by casting it as an assignment problem, as elaborated in Algorithm~\ref{algo:DPP-RA}.  
We show the near-optimal performance of the DPP-RA routine in the following theorem. We recall that $f_l^* (\mathcal{G})$ is the optimal objective value of \eqref{eq:RAproblem}. 

% and let $C = \frac{1}{2}\sum_{m=1}^{M} (P_{m,\text{avg}}^2 + P_{\rm max}^2+  B_{\rm max}^2 + A_{\rm max}^2)$. 

\begin{theorem}\label{thm:DPP-RA}
The Algorithm~\ref{algo:DPP-RA} solves \eqref{DPPGreedyOptimization} optimally. Moreover, the long-term expected average timely throughput of the DPP-RA, $f_{l}^{\rm DPP}(\mathcal{G})$ is bounded as 
\begin{align}\label{eq:bound}
f_l^* (\mathcal{G})- \frac{C}{V}\leq f_{l}^{\rm DPP}(\mathcal{G}) \leq f_l^* (\mathcal{G}),
\end{align}
and the average power constraint in \eqref{eq:averagePowerConstraint} is satisfied, for any $l\in \mathcal{L}$ and $\mathcal{G}\subseteq \mathcal{M}$ and $V>0$.
\end{theorem}
\begin{proof}
See Appendix B.  
%The proof that DPP algorithm provides an $O(\frac{1}{V})$-optimal solution to \eqref{eq:RAproblem} hinges on the fact that there exists a channel-only (that depends on the instantaneous realization of channel power gain vectors, $\mathbf{h}_m$) policy, $\omega^{*}$ for the problem defined in \eqref{eq:RAproblem}. We explain details of the proof in Appendix~\ref{apndx: appendix b}.
\end{proof}
%\begin{comment}
We  make  following remarks on Algorithm~\ref{algo:DPP-RA} and Theorem~\ref{thm:DPP-RA}. 
\begin{itemize}
	\item The complexity of directly solving \eqref{DPPGreedyOptimization} is at least   $(\min(U,M))!\times |\mathcal{P}|$, where $!$ represents the factorial operation. This is because, when $U> M$, the set of RUs can be allocated to the first STA in $U$ ways, to the second STA in $U-1$ ways and so on, and similarly, when  $M> U$, the set of STAs can be allocated to the first RU in $M$ ways, to the second RU in $M-1$ ways and so on. Moreover, for each RU-STA assignment, there are $|\mathcal{P}|$ possible power allocations. However, the complexity of Algorithm~\ref{algo:DPP-RA} is polynomial in $M$, $U$ and $|\mathcal{P}|$, as we first need to compute $U_B^*$ for each STA-RU assignment by searching over $|\mathcal{P}|$ values and then solve the assignment problem using the Hungarian algorithm which has a polynomial complexity. 
	\item 	By increasing $V$, the timely throughput of the DPP-RA routine can get arbitrarily close to that of the optimal policy. However, the time required for the average powers to converge to constants less than or equal the bounds for each STA will be longer with increasing $V$. 
\end{itemize}

\subsubsection{Round-Robin (RR) RA  Routine}\label{sec:Round Robin RA Algorithm}
In the RR-RA routine, we allocate transmit power independent of the RU allocation and channel power gain. The allocated power, $P_{m,k}(t)$ (for any $k$) is chosen to be the maximum power less than or equal to $P_{\rm max}$ that satisfies the average power constraint, $P_{m,\text{avg}}$ at each block of a TWT-SP.
Having allocated the power, the RUs are allocated to the STAs in a round-robin fashion, without accounting for the  underlying channel power gain or the data arrival random process.

\subsubsection{Greedy RA Routine} \label{eq:Greedy RA Algorithm}

 %In (\ref{eq:DPP RA Algorithm}), our goal was to provide queue stability in addition to maximize the total instantaneous timely-throughput of the network considered. 
In the Greedy RA routine, the power allocation is identical to what is adopted in the RR routine, as mentioned in Section~\ref{sec:Round Robin RA Algorithm}. Given the power allocation, in the frame transmission at the $t^{\rm th}$ block inside a TWT-SP of the $l^{\rm th}$ broadcast group, an instantaneous timely-throughput matrix is constructed, where $\left(m,k\right)^{\rm th}$ entry of the matrix is the instantaneous timely throughput of the $m^{\rm th}$ STA if it is allocated the  $k^{\rm th}$ RU and transmit power $P_{m,k}(t)$, $m \in \mathcal{G}_l$ and $k \in \mathcal{U}$, which can be computed from \eqref{eq:ImmediateTimelyThroughput}. Then the Hungarian algorithm is used to obtain the optimal RU allocation, $u_{m,k}(t)$ for all  $m\in \mathcal{M}$ and for all  $k\in \mathcal{U}$.

	\begin{algorithm}[b]
	\caption{The Greedy Grouping Routine.}
	\label{algo:Greedy-Grouping}
	{\small
		\begin{algorithmic}[1]
			\Procedure{Greedy-Grouping}{}
			\State Assign $\mathcal{S}_l=\phi$ for all $l\in \mathcal{L}$. 
			\While {$\cup_{l\in \mathcal{L}} \mathcal{S}_l \neq \mathcal{M}$}
			\State  Compute $\Delta_l(\mathcal{S}_l,m)$ for each $m \notin \cup_{l\in \mathcal{L}}\mathcal{S}_l$ and $l\in \mathcal{L}$. 
			\State Assign STA $m$ to group $l$ if $\Delta_l(\mathcal{S}_l,m) > \Delta_k(\mathcal{S}_k,n)$ for all $k\in \mathcal{L}\setminus \{l\}$ and $m,n\notin \cup_{l\in \mathcal{L}}\mathcal{S}_l$. 
			\State Update $\mathcal{S}_l$ for all $l\in \mathcal{L}$.
			\EndWhile
			\EndProcedure
	\end{algorithmic}}
\end{algorithm}
\subsection{Solution to the Grouping Problem in \eqref{eq:Groupingproblem}}\label{eq:GroupingSolution}
Below, we solve the grouping problem suboptimally.

% \subsubsection{Exhaustive Search Grouping}
% In this method, we compute the long-term expected average timely throughput for each possible way in which $M$ STAs can be grouped into $L$ groups and pick the grouping that yields the largest total timely throughput. We note that the total number of ways in which  $M$ STAs can be grouped in to $L$ disjoint groups, where some groups can be empty, is $\sum_{l=1}^{L}S(M,l)$, where $S(M,L)$ is called as Stirling number of second kind or Stirling partition number. For some $M$ and $L$ it may be prohibitive to obtain the optimal grouping and hence, we propose suboptimal simpler routines for grouping in the following. 
 \begin{table}[b]
 \caption{System parameters}
\centering
\begin{tabular}{ p{6 cm}|p{3 cm} }
 \hline
 \textbf{Description and Notation} & \textbf{Value} \\ [0.5ex] 
 \hline
 Number of STAs ($M$) & $8$  \\ 
Number of RUs ($U$) & $4$  \\
Center Frequency ($f_c$) & $2.4$ GHz  \\
Block Duration (FTT) &1 ms  \\ 
%Channel States ($\mathcal{H}$) & \{'Good', 'Medium', 'Bad'\} \\

%Pathloss exponent & $3.5$ \\
Maximum Transmit Power ($P_{\rm max})$ & $1$ Watt\\
Control Parameter of DPP algorithm ($V$) &$10000$ \\ 
Bernoulli Packet Generator parameters: 
& \\
(i) Number of packets per batch ($b$) &$10$ \\
(ii) Probability of batch generation ($p$) &$0.7$ \\ [1ex]
 \hline
\end{tabular}
\label{tab:table-1}
\end{table}

%Further, we initialise the following important broadcast TWT parameters which will be used throughout the simulation: 
 \begin{table}[b]
 \caption{TWT parameters}
\centering
\begin{tabular}{ p{6 cm}|p{2 cm}|p{2 cm}|p{2 cm} }
 \hline
 \textbf{Description and Notation} & \textbf{Group 1} &\textbf{Group 2} &\textbf{Group 3} \\ [0.5ex] 
 \hline
 TWT Offset (TWT-O) & $2$ ms & $16$ ms & $10$ ms  \\ 
 TWT Wake Interval (TWT-WI) & $30$ ms &  $150$ ms
 & $90$ ms\\
 TWT Service Periods (TWT-SP)  & $7$ ms & $2$ ms & $5$ ms  \\
 %UL OFDMA Frame Transmission Time & $4$ ms & $2$ ms  \\  [1ex] 
 \hline
\end{tabular}
\label{tab:BroadcastTWTParameters}
\end{table}
%One can group $M$ STAs into $L$ disjoint non-empty groups in $S(M,L)$ ways, where $S(M,L)$ is called as Stirling number of second kind or Stirling partition number. In our case, we allow for a group to be empty as long as the maximum number of groups is $L$. Hence,  

\begin{figure}[t]
		\begin{center}
			\includegraphics[scale = 0.75]{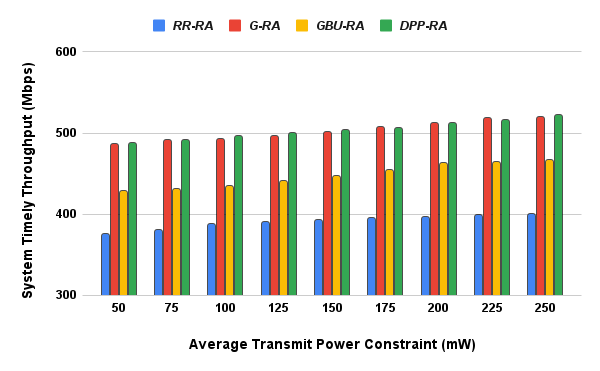}
		\end{center}
		\caption{System Timely Throughput versus Average Transmit Power constraint for fixed number of STAs. GBU-RA implies Greedy Buffer Unaware algorithm is used for resource allocation. Similar naming convention is used for other bars as well.}
		\label{fig:TT1}
\end{figure}
\subsubsection{Greedy Grouping Algorithm} \label{eq:Greedy Grouping Algorithm}
We use a greedy algorithm, proposed in \cite{DBLP:Submodular-maximization-journals/corr/abs-1810-12861}, for solving \eqref{eq:Groupingproblem}.  
% $\mathcal{G}_l$ has a monotone and submodular valuation function
Recall that the long-term average expected timely throughput obtained by solving 
\eqref{eq:RAproblem}  via a method described in the previous subsection for a group of STAs, $\mathcal{G}$, with TWT parameters, $o_l$, $\tau_l$ and $\zeta_l$, is denoted by  $f_l(\mathcal{G})$. Note that $f_l(\cdot): 2^\mathcal{M} \rightarrow \mathbb{R}$, where we define $f_l(\mathcal{\phi}) = 0$. We cast the grouping problem as a partitioning problem, where we need to find a partition of the set of STAs, $\mathcal{M}$, among $L$ groups such that the sum of the valuations of the groups (which is the weighted long-term expected average timely throughput of the STAs in the group, in our case)  is maximized.
The iterative greedy routine works as follows: Let $\mathcal{S}_l$ be the set of STAs allocated to group $l\in \mathcal{L}$ at a given iteration.  Let the increment in reward when STA $m\notin \cup_{l\in \mathcal{L}} \mathcal{S}_l$ is allocated to group $l$ given the existing confirmed set $\mathcal{S}_1,\ldots,\mathcal{S}_L$ be denoted by $\Delta_l(\mathcal{S}_l,m)$. That is,
\begin{align}\label{eq:Increment1}
\Delta_l(\mathcal{S}_l,m)\triangleq f_l(\mathcal{S}_l \cup \{m\}) - f_l(\mathcal{S}_l), \;\; \forall l\in \mathcal{L}.
\end{align}
An unallocated STA $m$ is then assigned to group $l$ if $\Delta_l(\mathcal{S}_l,m) > \Delta_k(\mathcal{S}_k,n)$ for all $k\in \mathcal{L}\setminus \{l\}$ and $m,n\notin \cup_{l\in \mathcal{L}}$.  This procedure is continued until all the STAs have been assigned to a group. 
We summarize this procedure in Algorithm~\ref{algo:Greedy-Grouping}. Since there are only a finite number of STAs, the algorithm will terminate in a finite number of steps.

\subsubsection{Round Robin (RR) Grouping Algorithm}\label{eq:Round Robin Grouping Algorithm}
In this grouping algorithm, the partition of the set of STAs, $\mathcal{M}$, among $L$ groups is performed as follows: We initialize all the groups as empty set i.e., $\mathcal{G}_l=0$, for all $l\in \mathcal{L}$. Then, starting with $l=1$, we assign STAs to group $l$ according to ascending order of their indices until $|\mathcal{G}_l|<\left\lfloor{({M+L-1})/{L}}\right\rfloor$. We then increment $l$ and repeat the previous step for each unallocated STA until all the STAs are exhausted.

\section{Numerical Results} 

 \begin{figure}[t]
		%\begin{center}
			\includegraphics[scale = 0.75]{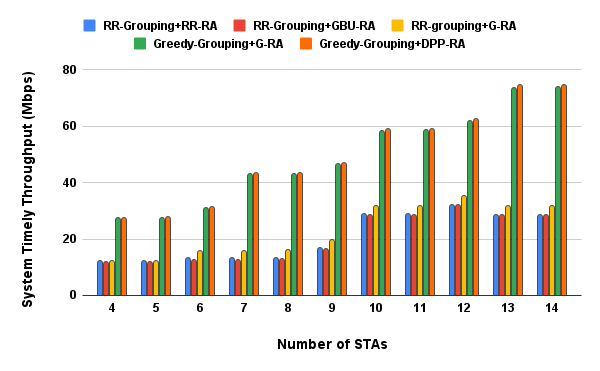}
		%\end{center}
		\caption{System Timely Throughput versus Number Of STAs for fixed Average Transmit Power Constraint. Greedy-Grouping+DPP-RA implies that the Greedy algorithm is used for grouping and the DPP algorithm is used for resource allocation. Similar naming convention is used for other bars as well.}
		\label{fig:SystemTimelyThroughputVsNumberOFSTAs}
\end{figure}

% \begin{figure}[t]
% 		\begin{center}
% 			\includegraphics[scale = 0.3]{Plot03AveragePacketLossWithM2to14L2TotalTXSlots1000000.png}
% 		\end{center}
% 		\caption{Average Packet Loss Rate versus versus Number Of STAs for fixed Maximum Transmit Power.}
% 		\label{fig:AverageOPacketLossVsNumberOfSTAs}
% \end{figure}

\balance 
 In this section, we present simulation results for an uplink communication scenario in an IEEE 802.11ax WLAN setup, in which $M$ STAs and an AP present inside a BSS support TWT, using the algorithms presented in the previous section. The channel power gain of an STA is assumed to be uniformly distributed over $\{10, 0.1, 0.001\}$ in a slot, which transitions independently across slots. 
%The available transmit power set is obtained by calculating the minimum transmit power required to transmit all possible number of packets for each channel states mentioned earlier.
An STA generates traffic from either a Bernoulli packet generator which assumes that data arrive at each slot as a \emph{batch} which contains a fixed number of packets, $b$ with probability of batch generation $p$ or from a realistic packet generator consisting of a Buffered Video streaming (BV) and two constant bit rate (CBR) traffic models, CBR-$1$ and CBR-$2$ with parameters mentioned in 
%WLAN TGax Evaluation Methodology document 
 \cite{WLAN-TGax-Evaluation-methodology}. In the BV traffic model, STAs generate frames at a rate of $30$ frames per second with frame size distributed as Weibull Distribution. The generated frames are further fragmented into packets of fixed size of $1500$ bytes. We control the offered data rate to $12$ Mbps by adjusting the parameters of the Weibull distribution. On the other hand, in the CBR traffic models, STAs generate data of fixed size of $3$ (CBR-$1$) and $40$ (CBR-$2$) Kilobytes at every fixed interval of time of $150$ (CBR-$1$) and $90$ (CBR-$2$) milliseconds producing a traffic of data rate $160$ Kbps (CBR-$1$) and $3.5$ Mbps (CBR-$2$). The associated deadlines for the frames generated by the BV, CBR-$1$ and CBR-$2$ traffic data models are  $30$, $150$ and $90$ milliseconds, respectively.

\subsection{Timely throughput variations of RA routines with, $P_{\rm avg}$} 
In Fig.~\ref{fig:TT1}, we present variations of system timely throughput when the RU and transmit power allocation occur according to the following RA routines: (i) Round-Robin (RR), (ii) Greedy (G), (iii) Greedy Buffer-Unaware (GBU), and (iv) Drift-Plus-Penalty (DPP), where the GBU routine is similar to the Greedy routine in Section~\ref{eq:Greedy RA Algorithm}, but it operates  without the knowledge of the instantaneous buffer status. Since the buffer status is unknown, one cannot compute the instantaneous timely throughput. Hence, we use the instantaneous throughput matrix for optimization using the Hungarian algorithm. 
For obtaining the results, we keep the number of STAs, $M$, and number of RUs, $U$ fixed and the average transmit power constraint, $P_{\rm avg}$ is varied. We assume STAs generate data according to Bernoulli packet generator. The system parameters used for simulation are mentioned in the Table \ref{tab:table-1}.

From the figure, we observe that the performance of DPP and Greedy RA routines is significantly better than the round-robin based benchmark algorithm (RR-RA). RA routines which are both channel-aware and buffer-aware, like DPP-RA and G-RA, perform better than only channel-aware and buffer-unaware routines, such as RR-RA, GBU-RA. 
 \subsection{Timely throughput variations with the number of STAs, $M$}
We now fix the average transmit power constraint, $P_{\rm avg}$ and vary the number of STAs, $M$. When number of STAs is increased, the new STA is considered to be generating data from one of the three realistic traffic models, BV, CBR-$1$ and CBR-$2$. The AP sets the TWT parameters as in the Table  \ref{tab:BroadcastTWTParameters} and broadcasts it to STAs. 
We show the obtained system timely throughput  in  Fig.~\ref{fig:SystemTimelyThroughputVsNumberOFSTAs}. From the figure, we observe that as the number of STAs increases, the system timely throughput increases in all the cases. Within the Round Robin based grouping (RR-Grouping) algorithm, RR-RA performs better than GBU-RA. This can be explained as follows: At a decision instant, when the number of STAs inside a broadcast group is less and the buffer information is unavailable, RA using only channel power gain information generates lower system timely-throughput than the benchmark RR-RA routine. 
%Via exhaustive Search Grouping algorithm, we obtain the optimum grouping of STAs into broadcast groups, which is close to the grouping obtained via Greedy Grouping Algorithm (See Numerical Results of \cite{longer-version}).
The performance of Greedy grouping algorithm along with the DPP and Greedy based RA routines continues to be better than the benchmark Round Robin based grouping along with benchmark RA routines. 

 \section{Conclusions}
%Power saving is a critical issue for portable mobile devices. In a dense communication network, where a large number of devices are present, channel contention poses major problem leading to packet re-transmissions, even packet discardments, causing poor QoE to STAs present in the network. The IEEE 802.11ax standard has introduced a multi-user scheduling and power saving mechanism called broadcast TWT, which enables negotiation between AP and group of STAs to decide the time at which they can transmit, allowing them to operate in a sleep mode at all other times. 
In this work, we considered a broadcast TWT set up based uplink  transmission and proposed a grouping mechanism that greedily partitions the STAs into a pre-defined number of broadcast groups and two resource allocation algorithms for allocating RUs and transmit powers to the STAs in each TWT group at the start  of a block inside a TWT-SP. We observed that the proposed grouping plus resource allocation algorithms, i.e., Greedy grouping with Greedy RA routine and Greedy grouping with Drift-Plus-Penalty based RA routine perform better than the baseline algorithms in terms of the considered system timely throughput metric. As part of the future work, we plan to analyse the complexity of the proposed grouping plus resource allocation algorithm. We also plan to investigate the actual performance benefit of the proposed algorithm in real WLAN system deployment.

\section*{Appendix}
\subsection{The Derivation of the Upper-Bound in \eqref{eq:DriftPlusPenalty}}\label{app:Upper-Bound}
We begin by deriving an upper bound on the conditional Lyapunov drift, $D(\Theta(t))$ defined in \eqref{eq:drift}. 
From \eqref{eq:virtualqueue}, we can write,
\begin{align}
	G_{m}^2(t+1) &= (\max\{G_{m}(t)-P_{m,\text{avg}}+P_{m}(t),0\})^2 \nonumber\\
	&\leq  G_{m}^2(t) + P_{m}^2(t) + P_{m,\text{avg}}^2 + 2 G_{m} (t) (P_{m}(t)- {P}_{m,\text{avg}}),\label{eq:powerInequality}
\end{align}
where the inequality in \eqref{eq:powerInequality} holds  because $	\left(\max\{a,0\}\right)^2 \leq a^2$. 
Similarly, we can expand the virtual queue defined in \eqref{eq:buffer evolution} as follows:
\begin{flalign}
&	B_{m}^2(t+1) = (\max\{\min(B_m(t)+a_m(t),B_{\text{max}})-R_m(t)-d_m(t),0\})^2 \nonumber\\
	&\leq (\min(B_m(t)+a_m(t),B_{\text{max}})-R_m(t)-d_m(t))^2 \label{firstInequality}\\
%	&= (\min(B_m(t)+a_m(t),B_{\text{max}})^2 + (R_{m}(t)+d_{m}(t))^2- 2(R_{m}(t)+d_{m}(t))(\min(B_m(t)+a_m(t),B_{\text{max}}) \\
	&\leq (\min(B_m(t)+a_m(t),B_{\text{max}})^2 + (R_{m}(t)+d_{m}(t))^2- 2(R_{m}(t)+d_{m}(t))B_{m}(t) \label{secondInequality} \\
	&\leq (B_m(t)+a_m(t)^2 + (R_{m}(t)+d_{m}(t))^2- 2(R_{m}(t)+d_{m}(t))B_{m}(t) \label{thirdInequality} \\
	&= B_{m}^2(t) + a_{m}^2(t) +(R_{m}(t)+d_{m}(t))^2+ 2B_{m}(t)(a_{m}(t)-R_{m}(t)-d_{m}(t)), 
\end{flalign}
where the inequality in \eqref{firstInequality} is because   $\left(\max\{a,0\}\right)^2 \leq a^2$, the inequality in \eqref{secondInequality} follows because $\left(\min\{B_{m}(t)+a_{m}(t),B_{\text{max}}\}\right) \geq B_{m}(t)$ and the inequality in \eqref{thirdInequality} holds true because $\left(\min\{a,b\}\right)^2 \leq a^2$.

Now consider the following:
\begin{align}
	&L(\Theta(t+1))  - L(\Theta(t)) = \frac{1}{2}\sum_{m=1}^{M} (B^2_m(t+1) + G^2_m(t+1) - B_m^2(t) - G_m^2(t) )\\
	&\leq \frac{1}{2}\sum_{m=1}^{M} P_{m,\text{avg}} ^2+ P_{m}^2(t) + 2 G_m (t) (P_m(t)- P_{m,\text{avg}}) + a_m^2(t)\nonumber\\
	&\;\;\;\;+  (R_m(t)+d_m(t))^2 
	+ 2 B_m (t) (a_m(t)- R_m(t)-d_m(t))\\
	& \leq C +  \sum_{m=1}^{M}   G_m (t) (P_m(t)- P_{m,\text{avg}})+     B_m (t) (a_m(t)- R_m(t)-d_m(t)),
\end{align}
for any $C \geq \frac{1}{2}\sum_{m=1}^{M} P_{m,\text{avg}}^2 + P_m^2(t)+  (R_m(t)+d_m(t))^2 + a_m^2(t)$.  
In this work, we choose, $C = \frac{1}{2}\sum_{m=1}^{M} (P_{m,\text{avg}}^2 + P_{\rm max}^2+  B_{\rm max}^2 + A_{\rm max}^2)$.  
%Based on the above bound,  we can obtain the  upper bound on the drift-plus-penalty expression mentioned in \eqref{eq:DriftPlusPenalty}. 
% \begin{align}\label{eq:DriftPlusPenaltyUpperBound}
% 	&D(\Theta(t)) - V\sum_{m=1}^{M}\mathbb{E}[R_m(t)|\Theta(t)]\nonumber\\  &\leq   C + \sum_{m=1}^{M} \left(\mathbb{E}\left[(G_m(t)(P_m(t)-P_{m,\text{avg}})|\Theta(t)\right]\right) \nonumber\\
% 	&+\sum_{m=1}^{M}\left(\mathbb{E}\left[(B_m(t))(a_m(t)-R_m(t)-d_m(t))|\Theta(t)\right]\right) \nonumber\\
% 	&-V\sum_{m=1}^{M}\mathbb{E}[R_m(t)|\Theta(t)]
% \end{align}

\subsection{Performance of the Proposed Drift-Plus-Penalty Algorithm}
\label{apndx: appendix b}
The optimality of the proposed algorithm in solving \eqref{DPPGreedyOptimization} follows because the Hungarian algorithm is optimal for assigning the RUs to STAs. For a given RU-STA assignment, any power value in $\mathcal{P}$ can be selected since doing so does not violate any of the constraints. Hence, the power that maximizes the objective function can be chosen for each decision instant. In the next, we prove the second part of the theorem. 
%We now prove the second part of the theorem. As described in Section~\ref{sec:DPP RA Algorithm}, the problem defined in \eqref{eq:RAproblem} can be transformed to a {series of} minimization problems as expressed in \eqref{DPP}, solved at the beginning of each time block \textit{t}. 

Let $\pi$ be the optimal policy obtained by Algorithm~\ref{algo:DPP-RA}. The objective function of \eqref{DPP} under policy $\pi$ can be bounded as follows: 
\begin{align}\label{eq:DPPinequality1}
&     D^{\pi}(\Theta(t)) - V\sum_{m=1}^{M}\mathbb{E}[R_{m}^{\pi}(t)|\Theta(t)]\nonumber\\
&\leq  C + \sum_{m=1}^{M} G_m(t)\mathbb{E}\left[P_{m}^{\pi}(t)-P_{m,\text{avg}}|\Theta(t)\right]  +\sum_{m=1}^{M} B_m(t)\mathbb{E}\left[a_{m}(t)-R_{m}^{\pi}(t)-d_{m}^{\pi}(t)|\Theta(t)\right]\nonumber\\
	&\;\;\;\;-V\sum_{m=1}^{M}\mathbb{E}[R_{m}^{\pi}(t)|\Theta(t)]\nonumber\\
 &\leq C + \sum_{m=1}^{M} G_m(t)\mathbb{E}\left[P_{m}^{\omega^*}(t)-P_{m,\text{avg}}|\Theta(t)\right]  +\sum_{m=1}^{M} B_m(t)\mathbb{E}\left[a_{m}(t)-R_{m}^{\omega^*}(t)-d_{m}^{\omega^*}(t)|\Theta(t)\right] \nonumber\\
	&\;\;\;\;-V\sum_{m=1}^{M}\mathbb{E}[R_{m}^{\omega^*}(t)|\Theta(t)],
\end{align}
where $\omega^*$ is the optimal packet-arrival and channel state-only policy for the original problem in \eqref{eq:RAproblem}. We note that our problem satisfies the boundedness conditions of Theorem 4.5 in \cite{series/synthesis/2010Neely} and hence there exists the optimal packet-arrival and channel state-only policy. In the above, the first inequality follows from the upper bound on the drift-plus-penalty expression and the second inequality follows because (i) the policy $\pi$ is the optimal policy that minimizes the upper bound on the DPP expression and (ii) $\omega^*$ may not be optimal in minimizing the upper bound on the DPP expression. 
Let the maximum value of the objective in \eqref{eq:RAproblem} be  $f^*$ and the  instantaneous value of quantity $x$ under policy $\omega^*$ be denoted by  $x^{\omega^*}(t)$ for $x\in \{f, P_m, R_m, d_m\}$. Concretely,  $f^{\omega*}(t) = \sum_{m=1}^{M}\mathbb{E} \left[R_{m}^{\omega^{*}}(t)\right]$.
In \cite{series/synthesis/2010Neely}, it is shown that for the policy $\omega^{*}$ for any $\delta>0$, the following inequalities hold true:
\begin{align*}
    & f^{\omega*}(t) \geq f^* - \delta, \\
    & \mathbb{E}\{P_{m}^{\omega^{*}}(t) - P_{m,\text{avg}}\} \leq \delta, \qquad \forall m \\
    & \mathbb{E}\{a_{m}(t)\} \leq \mathbb{E}\{R_{m}^{\omega^{*}}(t)+d_{m}^{\omega^{*}}(t)\} + \delta, \qquad \forall m
\end{align*}
for any $\delta>0$. By substituting the above inequalities into \eqref{eq:DPPinequality1}, we obtain
\begin{align*}
    D^{\pi}(\Theta(t)) - &V\sum_{m=1}^{M}\mathbb{E}[R_{m}^{\pi}(t)|\Theta(t)] \leq C - V(f^*-\delta)+  \sum_{m=1}^M(G_m(t)+B_m(t))\delta. 
\end{align*}
As $\delta \rightarrow 0$, due to the limit inequality theorem, we have, 
\begin{align*}
    D^{\pi}(\Theta(t)) - &V\sum_{m=1}^{M}\mathbb{E}[R_{m}^{\pi}(t)|\Theta(t)] \leq C - Vf^*,   
\end{align*}
and by taking iterated expectations and telescopic sum, we obtain, 
\begin{align}\label{eq:DPPinequality2}
    \mathbb{E}\left[L(\Theta(t))\right] - \mathbb{E}\left[L(\Theta(0))\right] - V\sum_{\tau=0}^{t-1}\sum_{m=1}^{M}\mathbb{E}\left[R_{m,\pi}(\tau)\right] \leq tC - Vtf^*
\end{align}

Now, we rearrange the terms in \eqref{eq:DPPinequality2}  and obtain: 
\begin{align*}
    V\sum_{\tau=0}^{t-1}\sum_{m=1}^{M}\mathbb{E}\left[R_{m,\pi}(\tau)\right] &\geq Vtf^*-\mathbb{E}\left[L(\Theta(0))\right] +  \mathbb{E}\left[L(\Theta(t))\right]-tC\nonumber\\
    &\geq Vtf^*-\mathbb{E}\left[L(\Theta(0))\right]-tC,  
\end{align*}
where the last inequality follows because $\mathbb{E}\left[L(\Theta(t))\right]$ is a positive quantity. 
We divide both sides of the above inequality by $tV$ and apply the limit $t\rightarrow \infty$ at both the sides of above inequality, which yields the following:
\begin{align*}
    \underset{t\rightarrow\infty}{\text{lim}}\frac{1}{t}\sum_{\tau=0}^{t-1}\sum_{m=1}^{M}\mathbb{E}\left[R_{m}^{\pi}(\tau)\right] \geq f^* - \frac{C}{V}. 
\end{align*}
 
In the below, we show that the optimal DPP policy, $\pi$ satisfies the average power constraints.  Note that the instantaneous timely throughput will be non-negative and bounded as the number of packet arrivals and resources are finite. Let $f^{\rm max}$ be the bound on the total instantaneous expected timely throughput, i.e.,  
\begin{align*}
    0 \leq \sum_{m=1}^{M}\mathbb{E}\left[R_{m}^{\pi}(\tau)\right] \leq f^{\max}, \qquad \tau \in \{0,1,\dots\}.
\end{align*}
From \eqref{eq:DPPinequality2}, we get
\begin{align*}
    \mathbb{E}\left[L(\Theta(t))\right] \leq \mathbb{E}\left[L(\Theta(0))\right] + Ct + Vt(f^{\max}-f^*). 
\end{align*}
Recall that 
\begin{align*}
    \mathbb{E}\left[L(\Theta(t))\right] = \frac{1}{2}\sum_{m=1}^{M}\mathbb{E}\left[B_{m}^{2}(t)\right] + \frac{1}{2}\sum_{m=1}^{M}\mathbb{E}\left[G_m^{2}(t)\right]
\end{align*}
Therefore, the following inequalities hold true for all $m$,
\begin{align*}
    &\mathbb{E}\left[B_m^{2}(t)\right] \leq 2\mathbb{E}\left[L(\Theta(0))\right]+2Ct + 2Vt(f^{\max}-f^*)   \\
    &\mathbb{E}\left[G_m^{2}(t)\right] \leq 2\mathbb{E}\left[L(\Theta(0))\right]+2Ct + 2Vt(f^{\max}-f^*)  
\end{align*}
As the variance of a random quantity is always greater than equal to zero, we have, $\mathbb{E}\left[B_m^{2}(t)\right]\geq \mathbb{E}^{2}\left[B_{m}(t)\right]$ and we can obtain
\begin{align*}
    \mathbb{E}\left[B_m(t)\right] \leq \sqrt{2\mathbb{E}\left[L(\Theta(0))\right]+2Ct+2Vt(f^{\max}-f^*)}. 
\end{align*}
We divide both sides of the inequality by $t$ and apply $t\rightarrow \infty$ with the assumption that $\mathbb{E}\left[L(\Theta(0))\right]\leq \infty$ and obtain 
\begin{align*}
    \underset{t\rightarrow \infty}{\text{lim}} \frac{\mathbb{E}\left[B_m(t)\right]}{t} = 0. 
\end{align*}
Along the above lines, we can also obtain $ \underset{t\rightarrow \infty}{\text{lim}} \frac{\mathbb{E}\left[G_m(t)\right]}{t} = 0$.
Hence, the queues defined in our problem are mean rate stable. The mean rate stability of virtual queues, $G_1(t), \ldots, G_M(t)$ implies that the average power constraints at all the STAs are satisfied.

\bibliographystyle{IEEEtran}
\bibliography{main}

\end{document}